\documentclass[preprint,12pt,authoryear]{elsarticle}
\usepackage{mydef}

\begin{document}
\begin{frontmatter}

\title{Nonlinear PDEs risen when solving some optimization problems in finance, and their solutions\tnoteref{t1}}
\tnotetext[t1]{The views represented herein are the author own views and do
not necessarily represent the views of New York University.}

\author[nyu]{A.~Itkin\corref{cor1}}
\ead{aitkin@nyu.edu}
\cortext[cor1]{Corresponding author}

\address[nyu]{Polytechnic School of Engineering, New York University, \\
6 Metro Tech Center, RH 517E, Brooklyn NY 11201, USA}

\begin{abstract}
We consider a specific type of nonlinear partial differential equations (PDE) that appear in mathematical finance as the result of solving some optimization problems. We review some existing in the literature examples of such problems, and discuss the properties of these PDEs. We also demonstrate how to solve them numerically in a general case, and analytically in some particular case. 
\end{abstract}


\end{frontmatter}

\section{Motivation and Introduction}

In this paper we consider a special type of the nonlinear PDE that arises in mathematical finance by applying optimization to some financial problems. To make it clear we start with the example given by  \cite{Lipton2001} who considered an optimal consumption problem, first formulated in \cite{merton:71}. 

\subsection{Lipton example}
In \cite{Lipton2001} Lipton considers a portfolio of domestic and foreign bonds. The relative value of the domestic bond $Y_t^0$ is constant, i.e. $Y_t^0 = Y_0^0 = 1$, while that of the foreign bond $Y_t^1$ is random and follows the Geometric Brownian motion with constant real-world drift $\mu$ and volatility $\sigma$, and $t$ is the time. Suppose $\Upsilon_t = (\Upsilon_t^0, \Upsilon_t^1)$ is a predictable self-financing trading strategy with the corresponding linear wealth $\Pi_t = \Upsilon_t^0 + \Upsilon_t^1 Y_t^1$, and $\omega_t^1 = \Upsilon_t^1 Y_t^1/\Pi_t$ is the fraction of this wealth invested in foreign bonds. Self-finance of the strategy means that given $\Pi_0$ one has 
\begin{equation*}
d\Pi_t = \Upsilon_t^1 dY_t^1 = \Upsilon_t^1(\mu dt + \sigma d W_t) = \omega^1_t \Pi_t (\mu dt + \sigma d W_t),
\end{equation*}
\noindent where $W_t$ is the Brownian Motion. The optimal consumption problem consists in finding the strategy $\omega_t^1$ which provides the expected wealth $\E0(\Pi_T)$ at maturity $T$ such that $\E0(\Pi_T) > \Pi_0$, while minimizing the variance of $\Pi_T$, i.e. $\E0(\Pi^2_t)$, where $\E0$ is the expectation conditional to the initial moment $t=0$. By using the Lagrange multiplier $\lambda$ and the function $J(t,Y^1,\Pi) = \min_{\omega^1} \left[\E0(\Pi_T^2 - \lambda \Pi_T)\right]$, the previous statement is equivalent to the solution of the Hamilton-Jacobi-Bellman (HJB) problem
\begin{align} \label{HJB}
\min_{\omega^1} & \left\{J_t + \dfrac{1}{2} \sigma^2 F(\omega^1) + \mu Y^1 J_{Y^1} + \omega^1 \mu \Pi J_{\Pi}\right\} = 0, \\
F(\omega^1) &= (Y^1)^2 J_{Y^1,Y^1} + 2 \omega^1 Y^1 \Pi J_{Y^1,\Pi} + (\omega^1)^2 \Pi^2 J_{\Pi,\Pi}. \nonumber
\end{align}

As at $t=0$ function $J(0,Y^1,\Pi) = \Pi_0^2 - \lambda \Pi_0$ doesn't depend on $Y^1$, so does the 
solution $J(t,Y^1,\Pi) = J(t,\Pi)$, and thus \eqref{HJB} reduces to 
\begin{equation} \label{HJB1}
\min_{\omega^1} \left[J_\xi - \dfrac{1}{2} (\omega^1)^2 \Pi^2 J_{\Pi,\Pi} - \omega^1 \bar{\mu} \Pi J_{\Pi}\right],
\end{equation}
\noindent where $\xi = \sigma^2(T-t), \bar{\mu} = \mu/\sigma^2$. Assuming $J_{\Pi,\Pi} > 0$ (see \cite{Lipton2001} for further discussion), \eqref{HJB} has an explicit solution $\omega^1 = - \bar{\mu} J_\Pi/(\Pi J_{\Pi,\Pi})$ that can be substituted back to \eqref{HJB} to give rise to the following problem for $J$:
\begin{equation} \label{NL1}
J_\xi = \dfrac{1}{2}\bar{\mu}^2 \dfrac{J^2_\Pi}{J_{\Pi,\Pi} }.
\end{equation}
This is a nonlinear equation of our interest. Based on the initial condition $J(0,\Pi) = \Pi^2 - \lambda\Pi$ it can be solved analytically by representing $J$ as a quadratic form of $\Pi$ with the coefficients being a function of $\xi$. Further details can be found in \cite{Lipton2001}.

\subsection{Carr example}

In \cite{CarrQUSA2014} Carr considers the call option price $C(S,t)$, with $S$ being the underlying spot price, to be the result of a maximization of some function, $F(p; S,t)$, with respect to some variable, $p$, holding the parameters $S$ and $t$ constant. He shows that if the call value $C(S,t)$ is a convex function of $S$, then $F(p;S,t) = p S - C^*(p;t)$, where $C^*(p;t)$ is the Legendre-Fenchel transform of the call value $C(S;t)$, namely: $C^*(p;t) \equiv \sup_{S} [p S - C(S;t)]$. He also shows that if the call value $C(S,t)$ is a convex differentiable function of $S$, the argument $p$ of the function $F$ is also the argument $p$ of the Legendre-Fenchel transform $C^*(p;t)$ of $C(S;t)$, and is also the delta of the call, i.e. $p = \fp{}{S} C(S;t)$. This result motivates his dual approach to option pricing. That means that in the standard (primal) approach to option pricing, the spot price process $S$ is an input to the approach, while the stochastic process followed by the delta of the option is an output of the approach. In contrast, the Carr's dual approach begins by specifying a stochastic process for the delta $p$ of a call, rather than for the spot price $S$ of the underlying asset. 

Further on, Carr considers a forward contract on the same underlying and assumes that the value $X_t$ of the forward contract solves the stochastic differential equation:
\begin{equation} \label{abm}
dX_t = a(X_t ,t) dW_t, \qquad t \in [0,T),
\end{equation}
\noindent where $a(X_t ,t)$ is a local volatility function. In other words, $X_t$ follows an arithmetic Brownian motion. Then the call value function $c(x,t) \equiv \EQ [X^+_T|X_t = x]$ solves the following parabolic PDE:
\be \label{Cpde}
\frac{1}{2} a^2\left( x, t \right)  c_{11}(x,t) +  c_2(x,t) = 0, \qquad x \in \mathbb{R}, t \in [0,T).
\ee
subject to the terminal condition:
\be \label{Ctc}
c(x,T) = x^+, \qquad x \in \mathbb{R}.
\ee
Here $\EQ$ is the expectation under the risk-neutral measure, sub-index $_1$ in $c_1$ in Carr's notation means the first derivative on the first argument of $c$, and subindex $_2$ denotes same with respect to the second argument. Accordingly, $c_{11}$ is the second derivative on the first argument.
 
Alternatively to the primal approach, in the dual approach Carr introduces $x$ to be now a function of $(p,t)$, so the option price in these variables $c^*(p,t) $ is
\be
c^*(p,t) \equiv p \cdot x(p,t) - c(x(p,t),t), \qquad p \in (0,1), t \in [0,T)
\label{lt}
\ee
Accordingly, under some mild assumptions the first derivative of $c^*(p,t)$ on the first argument is
\be \label{ltd}
c^*_1(p,t) =  x(p,t) + p x_1(p,t) - c_1(x(p,t),t) x_1(p,t) = x(p,t), 
\ee
\noindent since $p = c_1(x(p,t),t)$, and 
\be \label{ltd1a}
c^*_{11}(p,t) =  x_1(p,t) = \frac{1}{c_{11}(x(p,t),t)}, \qquad p \in [0,1], t \in [0,T),
\ee
\noindent by the inverse function theorem. Since $c_{11}(x,t) \geq 0$ implies $c^*_{11}(p,t)\geq 0$, $c^*$ is convex in $p$. Also based on the properties of the Legendre transform $x = c^*_1(p,t)$, and thus
\be \label{ilt}
c(x,t) = x c^*(p(x,t),t) - c^*(p(x,t),t), \qquad x \in \mathbb{R}, t \in [0,T).
\ee

Using It\^{o}'s lemma, Carr also derives an important relationship between $v(p,t)$ - the function relating the volatility of the delta process $P_t \equiv c_1(X_t,t), t\in [0,T)$ to the level of the call delta $p \in (0,1)$  and to time $t \in [0,T)$:
\be \label{ddg}
v^2(p,t) c^*_{11}(p,t)  = c_{11}(x(p,t),t) a^2(x(p,t),t), \qquad p \in (0,1), t \in [0,T).
\ee
\noindent and finally a dual PDE:
\be \label{Cpded}
\frac{1}{2} v^2\left(p, t \right)  c^*_{11}(p,t) -  c^*_2(p,t) = 0, \qquad p \in (0,1), t \in [0,T),
\ee
\noindent subject to the terminal condition
\be \label{tcc2}
\lim\limits_{t \uparrow T} c^*(p,t) =  x 1(x >0) - x^+ = 0.
\ee
While the standard PDE in \eqref {Cpde} is commonly referred to as a backward equation, the dual PDE in  \eqref {Cpded} is a forward equation.

Despite this was not discovered in \cite{CarrQUSA2014}, we will show how to translate \eqref{Cpded} to the nonlinear PDE.

\subsection{Further notes}
The above consideration implicitly assumes that the local volatility function $v(p,t)$ of the delta process is known somehow for all necessary values of $p$ and $t$. And then the forward PDE \eqref{Cpded} with the terminal condition \eqref{tcc2} solves the problem. From a practical point of view this assumption, however, could be unrealistic in a sense that the market doesn't provide $v(p,t)$ explicitly. At the first glance one of the possible approaches to overcome this could be first to solve the backward problem, and then use the condition \eqref{ddg} to restore $v(p,t)$ from the precomputed values of the option gammas $c_{11}(x,t)$ given a map $p \to p(x,t)$. However, this is not fully correct. Indeed, suppose that we solved the backward problem, e.g., using a finite difference method, and now have $c_{11}(x,t)$ at every point on the two-dimensional grid in $(x,t)$. Hence, at $t=0$ we have the values of $p = p(x,0)$ on this grid. Assume that we now want to solve the forward problem, and use these values of $p$ as our $p$ grid. Further on, let us make one step in time forward with a time step $\Delta t$. Obviously, $p(x,\Delta t) \ne p(x,0)$.  Therefore, if we still use the $p$-grid constructed by using the values $p(x,0)$ with the $x$-grid fixed at $t=0$, the inverse map $x(\Delta t) = x(p, \Delta t)$ doesn't coincide with $x(0)$. As a consequence, the values of $a(x.t)$ and $c_{11}(x,t)$ defined at the $x$ grid has to be re-interpolated at time $\Delta t$ onto the  $x(\Delta t)$ grid.

Even worse, for doing that we need to know $x(\Delta t)$ which correspond to the $p$ grid defined at $t=0$. So either we need also to store in memory all values of  $c_{1}(x,t)$ during the backward recursion, or use the connection $x(p,t) = c^*_{1}(p,t)$. In the later case the forward linear equation \eqref{Cpded} becomes nonlinear, because based on \eqref{ddg} we need to use the map $x(p,t)$ where $x$, in turn, is a function of $c^*$. Below when presenting the numerical results we will discuss the peculiarities of the map $x(p,t)$ in more detail.

Also it is worth mentioning that this approach, in general, has some pitfalls. First, once the backward problem is solved, most likely we don't need the solution of the dual problem. A possible exclusion is related to valuation of the long-dated options. Second, storing all values of gammas (and perhaps deltas) on the grid in $x$ and $t$ could be memory consuming. Therefore, we present an alternative view to this problem.

In doing that below we relax the above assumption that the local volatility function $v(p,t)$ of $P_t$ is known. Instead we assume that the only local volatility $a(x,t)$ of the forward process $X_t$ is known. We express it in the new variables $p,t$, and to eliminate any confusion introduce new notation for thus expressed function $a$, namely $\hat{a}(p,t)$. Note, that in order to build a map $(x,t) \to (p(x,t),t)$ we still need to know the option values $c(x,t)$ and deltas $c_1(x,t)$ at $t=0$ which depend on the option expiration $T$.

As per this assumption $v(p,t)$ is now not an independent input of the model. Rather, it is determined by \eqref{ddg} where $c_{11}(x(p,t),t)$, in turn, connects to $c^*_{11}(p,t)$ by \eqref{ltd1a}. Collecting all these expressions together and rearranging terms we arrive at the following equation.
\begin{equation} \label{nonl}
c^*_2(p,t) c^*_{11}(p,t) = \frac{1}{2} \hat{a}^2\left(c^*_1(p,t), t \right), \qquad p \in (0,1), t \in [0,T).
\end{equation}

This is a nonlinear equation with the terminal condition given by \eqref{tcc2}. To the best of our knowledge this type of equations has not been considered in the theory of the nonlinear PDEs as well as in mathematical finance with the exception of the Lipton case considered earlier in this section.

Despite a lack of the initial condition, this is still a forward equation in a sense that the calendar time $t$ runs from 0 to $T$. However, because of the terminal condition, this problem is not the standard Cauchy boundary problem. In physics this type of problems is well-known and usually is solved using either invariant imbedding, \cite{BellmanWing1992}, or the shooting method, \cite{sm}.

Therefore, the main contributions of this paper are i) recognition that the Carr's dual approach gives rise to a certain type of non-linear PDEs, that in the simplest form of the local volatility function linear in the underlying process have a predecessors, see \cite{Lipton2001}, but in general are more complex; and ii) an algorithm of how to efficiently solve these equations using a finite-difference method.

The rest of the paper is organized as follows. In the next section we make analysis of the derived nonlinear PDE and also provide an example (a Black-Scholes setting) when it can be solved analytically. In section~\ref{NM1} we propose a new numerical method to solve the full nonlinear equation which is unconditionally stable and of the second order of approximation in space and time. In section~\ref{NE1}
we present some numerical experiments and discuss their results. The final section concludes and provides suggestions for the future research.

\section{Analysis of the new nonlinear PDE \label{nonlPDE}}

A direct analysis shows that despite an unusual form \eqref{nonl} is well behaved. Indeed, the Legendre transform's Gamma $c^*_{11}(p,t)$ is nonnegative, therefore the Legendre transform's theta should also be nonnegative. As $c^*(p,t) < 0$ this means that when $t$ runs forward the values of $c^*(p,t)$ decrease in the absolute value and in the limit $t \to T$ reach their correct terminal value $c^*(p,T) = 0$. On the other hand, as $c_{11}(x,t)$ has a bell-shaped profile, $c^*_{11}(p,t)$ as per \eqref{ltd1a} has an inverse bell-shaped profile. In other words, for deep ITM and OTM options $c^*_{11}(p,t)$ tends to infinity, while close to ATM it is positive and close to zero. That means that $c^*_2(p,t)$ is zero at $p=0$ and $p=1$ and is positive at the intermediate values of $p$. So the evolution of $c^*(p,t)$ should be that as presented in Fig.~\ref {FigForw}.

It turns out that \eqref{nonl} also admits a dual representation. To show this, we start with \eqref{Cpded},
take into account \eqref{ltd1a} and also the equation
\be \label{ltd2}
c^*_2(p,t) =  p x_2(p,t) - c_1(x(p,t),t) x_2(p,t) - c_2(x(p,t),t) = - c_2(x(p,t),t),
\ee
\noindent derived in \cite{CarrQUSA2014} (since $p = c_1(x(p,t),t)$). Also use the definition of Delta to be $p = c_1(x,t)$ and make a change of time from forward $t$ to backward $\tau = T-t$. Thus transformed \eqref{nonl} now reads
\begin{equation} \label{nonlDual}
c_2(x,\tau) c_{11}(x,\tau) = \frac{1}{2} v^2\left(c_1(x,t), \tau \right), \qquad x \in [0,\infty), \tau \in [0,T).
\end{equation}

\subsubsection{Example.} Consider a particular case when the local volatility function is $a(x,t) = \sigma x$. In other words this corresponds to the displaced lognormal model
\begin{equation} \label{displ}
dS_t = \sigma(S_t-K) d W_t
\end{equation}
\noindent where $K$ is the option strike. Substituting this definition of $a(x,t)$ into \eqref{nonl} we arrive at the following nonlinear equation
\begin{equation} \label{nonlLipton}
c^*_2(p,t) = \frac{1}{2} \sigma^2 \dfrac{(c^*_1)^2}{c^*_{11}(p,t)}.
\end{equation}
One can easily see that this is exactly \eqref{NL1} despite with subject to different boundary and initial conditions. We remind that in \cite{Lipton2001} \eqref{nonlLipton} was solved by assuming that (in this section notation) $c^*(p,t)$ is a quadratic function of $p$. However, in our case this ansatz doesn't obey the boundary and terminal conditions for $c^*(p,t)$. Nevertheless, \eqref{nonlLipton} admits a semi-explicit closed form solution.

To obtain this solution, we re-write the definition of $c^*(p,t)$ in \eqref{lt} in the form $c^*(p(x,t),t) = x p(x,t) - c(x,t)$. As we assumed that $a(x,t) = \sigma x$, $X_t$ obeys the following stochastic SDE (Geometric Brownian motion)
\[ dX_t = \sigma X_t dW_t, \qquad -K < X_t < \infty. \]
By the Feynman-Kac theorem $c(x,t)$ solves 
\begin{equation} \label{FK}
c_t(x,t) = \dfrac{1}{2}\sigma^2 x^2 \sop{c(x,t)}{x}, 
\end{equation}
\noindent subject to the following conditions
\[ c(x,T) = x^+, \quad c(-K,t) = 0, \quad \quad c(x,t) \to x \ \mbox{at} \ x \to \infty, \quad -K \le x < \infty. \]

The explicit solution of \eqref{FK} can be obtained as a sum of two branches. 

At $x > 0$ the solution is 
$c(x,t) = x$. Indeed, this expression solves \eqref{FK} as well as obeys the terminal condition and the boundary condition at infinity.

At $x \le 0$ \eqref{FK} gives the price of the up-and-out call with the upper barrier at $K$ and the strike $K=0$ written on the negative process $X_t$ (so $-X_t \ge 0$). Therefore, it reads, see \cite{Howison1995}
\begin{align} \label{how}
c(x,t) &= c_{BS}(y,t,0) - c_{BS}(y,t,K) - K c_{d}(y,t,K) \\
&- \left(\dfrac{y}{K}\right)^{2 \alpha} \left[
c_{BS}(z,t,0) - c_{BS}(z,t,K) - K c_{d}(z,t,K) \right]. \nonumber
\end{align}
Here $y = -x,\ z = -K^2/x$, $c_{BS}(S,t,K)$ is the Black-Scholes call option price, $c_d(S,T,K) = e^{r(T-t)}N(d_1(S,T,K))$, $\alpha = \frac{1}{2}(1-k'), \ k' = 2(r-q)/\sigma^2$. In our case $r=q=0, \ \alpha = 1/2, \ c_{BS}(y,t,0) = c_{BS}(z,t,0) = x$, and, therefore, \eqref{how} simplifies to 
\begin{align} \label{how1}
c(x,t) &= y - c_{BS}(y,t,K) - K N(d_1(y,t,K)) \\
&- \dfrac{y}{K}\left[ c_{BS}(z,t,0) - c_{BS}(z,t,K) - K N(d_1(z,t,K))\right]. \nonumber
\end{align}
As at $x = -K$ we have $y = z = K$, this solution provides the correct boundary value of $c(-K,t) = 0$. Also at $t=T$ we have $c(x,T) = 0$.

Thus, these two branches of $c(x,t)$ solve \eqref{nonlLipton}, provide the correct boundary values, and obey the terminal condition. What remains to prove is that this solution is continuous at $x=0$. Indeed, the first branch at $x > 0$ in the limit $x \to 0$ provides $c(0,t) = 0$. The second branch at $x < 0$ in the limit $x \to 0$ also provides $c(x,t) \to 0$. The latter is not obvious because $c_{BS}(z,t,0) = K^2/y$, and it seems that the first term in the square brackets in \eqref{how1} multiplied by $y/K$ doesn't converge to 0 at $y \to 0$. However, one has to have in mind, that \eqref{how1} is obtained by using the method of images. And it could be shown that at $x \to 0$ the second line in \eqref{how1} vanishes regardless of the value in the square brackets, because the multiplier $y/K$ vanishes. 

Thus, we managed to solve \eqref{FK} in closed form. Now by \eqref{lt} 
\begin{equation} \label{cstar}
c^*(x,t) = x c_1(x,t) - c(x,t) = 
\begin{cases}
0, & x \ge 0 \cr
\bar{c}^*(x,t), & x < 0
\end{cases}
\end{equation}
\noindent where 
\begin{align*}
\bar{c}^*(x,t) &= K \left[ N(d_1(y,t,K)) - N(d_2(y,t,K)) + \dfrac{\phi(y,t,K)}{\sigma \sqrt{t}} \right] \\
&- y \left[ N(d_1(z,t,K)) - N(d_2(z,t,K)) - \dfrac{\phi(z,t,K)}{\sigma \sqrt{t}}  - 2\dfrac{K}{y} 1_{y \ne 0}\right], \\
\phi(x) &= \dfrac{e^{-x^2/2}}{\sqrt{2}}.
\end{align*}

To complete the solution we need to take into account that $x = x(p,t)$. To obtain this dependence in closed form, observe that 
\[ x(p,t) = c^*_1(p,t) = c^*_1(x,t) \fp{x}{p} \]

From this expression
\begin{equation} \label{xp}
p = \int \dfrac{c_1^*(x,t)}{x} dx = \dfrac{c^*(x,t)}{x} + \int \dfrac{c^*(x,t)}{x^2} dx.
\end{equation}
Thus 
\[ p = \mbox{Inv}\left[\dfrac{c^*(x,t)}{x} + \int \dfrac{c^*(x,t)}{x^2} dx \right] \]
\noindent where $\mbox{Inv}$ means the inverse function. Having an explicit representation for $c^*(x,t)$ this map $x = x(p,t)$ can be computed numerically.

Aside of this example the obvious question is: how to solve \eqref{nonl} with allowance for all its peculiarities and non-linearity. For doing that below we propose an iterative numerical scheme which assumes that the initial value $c^*(p,0)$ is somehow known. 

\section{Numerical method \label{NM1}}

As our intention is to solve \eqref{nonl} numerically, and, moreover, iteratively, we introduce index $k$ which counts the current iteration number. Index $k$ starts from 1 and runs up to a certain number at which the scheme converges in some norm. Also for simplicity of notation in this section we omit superscript $^*$ in the definition of the Legendre transformed option price $c^*(p,t)$, and use $C$ instead of $c$. Thus, below we denote it as $C(p,t)$.

We start by re-writing \eqref{nonl} in the form
\begin{equation} \label{nonl1}
c^*_2(p,t) = \frac{1}{2} \mathcal{D}(p,t) c^*_{11}(p,t), \qquad p \in (0,1), t \in [0,T).
\end{equation}
\noindent where 
\[ \mathcal{D}(p,t) =  \dfrac{1}{2}\left[\dfrac{\hat{a}(c^*_1(p,t),t}{C_{11}(p,t)}\right]^2, \]
If $\mathcal{D}(p,t)$ is a function of $p$ and $t$ only, \eqref{nonl1} would be a linear heat equation. However, in our case $\mathcal{D}(p,t)$ itself is a function of $c^*(p,t)$, and even worse - a function of its derivative $C_{11}(p,t)$.
Nevertheless, the discrete solution of the \eqref{nonl1} could be represented in the operator form
\begin{equation} \label{Scheme}
C(p,t+\Delta t) = \exp\left\{\dfrac{1}{2} \Delta t \mathcal{D}(p,t) \triangle \right\} C(p,t),
\end{equation}
\noindent where $\triangle$ is a finite difference approximation of the second derivative $\partial_{pp}$ on the grid. For instance, on a uniform grid in $p$ with step $h$ 
\[ \triangle C(p,t) = \dfrac{1}{h^2}\left[C(p+h,t) - 2 C(p,t) + C(p-h,t)\right] \]
which approximates $C_{11}(p,t)$ with $O(h^2)$. By using Taylor series expansion of \eqref{Scheme} it could be validated that this scheme approximates \eqref{nonl} with $O(h^2)$ and is exact in $\Delta
t$.

Note that again, using the Taylor series expansion one can verify that the following scheme 
\begin{equation} \label{Scheme2}
C(p,t+\Delta t) = \exp\left\{\dfrac{1}{2} \Delta t D(p,t+\Delta t) \triangle \right\} C(p,t)
\end{equation}
\noindent also provides same order of approximation of \eqref{nonl} in $h$. Combining them together we finally propose to use a scheme
\begin{equation} \label{SchemeF}
C(p,t+\Delta t) = \exp\left\{\dfrac{1}{4} \Delta t \left[ D(p,t) + D(p,t+\Delta t)\right] \triangle \right\} C(p,t)
\end{equation}

Now we setup an iterative algorithm to solve this discrete nonlinear equation. 
\begin{enumerate}
\item At the first iteration we re-write \eqref{Scheme} in the form
\begin{equation} \label{k1}
C^{(1)}(p,t+\Delta t) = \exp\left\{\dfrac{1}{2} \Delta t D(p,t) \triangle \right\} C(p,t)
\end{equation}
In the expression $C^{(k)}$ the superscript $^{(k)}$ marks the value of $C$ found at the $k$-th iteration of the numerical procedure. In other words, at the first iteration we set $C^{(1)}(p,t+\Delta t) = C^{(1)}(p,t)$ but use this substitution only to compute $D(p,t+\Delta t)$ in the RHS of \eqref{SchemeF}. The \eqref{k1} could be solved by either computing the discrete matrix exponential given the grid in $p$ (which could be done with complexity $O(N^2)$, $N$ be the number of the grid points), or by using any sort of P\'ade approximation of the exponent. For instance, P\'ade approximation $(1,1)$ leads to the well-known Crank-Nicholson scheme which could be solved with the linear complexity in $N$.

\item To proceed we represent \eqref{SchemeF} in the form
\begin{equation} \label{kExp}
C^{(k)}(p,t+\Delta t) = \exp\left\{\dfrac{1}{4} \Delta t \left[ D(p,t) + D^{(k-1)}(p,t+\Delta t) \right] \triangle \right\} C^{(k-1)}(p,t+\Delta t)
\end{equation}
So the next approximation $C^{(2)}(p,t+\Delta t)$ to $C(p,t+\Delta t)$ can be found by again either computing the matrix exponential, or by using some P\'ade approximation. 
Having this machinery in hands we can proceed in the same manner until the entire procedure converges, i.e. the condition
\[ \| C^{(k)}(p,t+\Delta t) - C^{(k-1)}(p,t+\Delta t)\| < \varepsilon \]
is reached after $k$ iterations with $\varepsilon$ being the method tolerance.
\end{enumerate}

To apply this algorithm we construct an appropriate discrete nonuniform grid ${\bf G}(x)$ in the delta space. For the sake of concreteness let \eqref{nonl} be solved at the space domain $[p_0,p_m]: \ p_0 =  0$, $p_m = 1$ using a nonuniform grid with $N+1$ nodes ($p_0, p_1,...,p_N$) and space steps $h_1 = p_1-p_0, ..., h_N = p_N - p_{N-1}$. We define a central difference approximation of the second derivative, see, e.g., \cite{HoutFoulon2010}
\begin{equation} \label{app2d}
C_{11}(p_i,t) = \delta_{i,-1}C_{11}(p_{i-1},t) + \delta_{i,0}C_{11}(p_{i},t) + \delta_{i,1}C_{11}(p_{i+1},t)
\end{equation}
\noindent where
\[
\delta_{i,-1} = \dfrac{2}{h_i (h_i + h_{i+1})}, \qquad
\delta_{i,0} = -\dfrac{2}{h_i h_{i+1}}, \qquad
\delta_{i,1} = \dfrac{2}{h_{i+1} (h_i + h_{i+1})}.
\]

\begin{proposition} \label{Prop1}
The scheme \eqref{kExp} with approximation \eqref{app2d} is a) unconditionally stable, b) preserves negativity of the solution, c) provides second order approximation of \eqref{nonl} on the grid $\mathbb{G}$, and d) converges.
\end{proposition}
\begin{proof}
See Appendix.
\end{proof}

\section{Numerical experiments \label{NE1}}
In the first test we solve the problem \eqref{Cpded} assuming that the local volatility function $v(p,t)$ is known. For doing that we first solve the backward problem using the following parameters of the model
$x=100, K=100, T =1$ yrs. Our grid contained $N=400$ nodes in $x$ and 100 steps in time. We used the local volatility function $a(x,t)$ given in Table~\ref {locVol1}
\begin{table}[!ht]
\begin{center}
\begin{tabular}{|c|c|c|c|c|c|c|c|c|c|c|c|c|}
\hline
$t, \mbox{yrs}$ & \multicolumn{9}{|c|}{$x$} \\ \cline{2-10}
& 70	& 80 & 90 &	100 & 110 &	120	& 130 &	140 & 150 \\
\hline
0.1 & 0.447 & 0.455 & 0.459 & 0.462 & 0.465 & 0.467 & 0.468 & 0.470 & 0.471\\
\hline
0.2 & 0.500 & 0.507 & 0.511 & 0.514 & 0.516 & 0.518 & 0.519 & 0.520 & 0.522\\
\hline
0.4 & 0.548 & 0.554 & 0.558 & 0.560 & 0.562 & 0.564 & 0.565 & 0.566 & 0.567\\
\hline
0.6 & 0.592 & 0.597 & 0.601 & 0.603 & 0.605 & 0.607 & 0.608 & 0.609 & 0.610\\
\hline
0.8 & 0.632 & 0.638 & 0.641 & 0.643 & 0.645 & 0.646 & 0.648 & 0.649 & 0.650\\
\hline
\end{tabular}
\caption{Local volatility function $a(x,t)$ of $X_t$.}
\label{locVol1}
\end{center}
\end{table}


The non-uniform grid in $x$ was constructed similar to \cite{ItkinCarrBarrierR3}. We solved the backward equation to obtain the option prices and deltas on the grid at $t=0$, and then used such obtained deltas $c_{1}(x,t)$ as a $p$ grid. This simultaneously provided us with the map $x \to x(p,0)$

Using thus obtained values of $c(p,t)$ we computed the Legendre transform $c^*(p,t)$ at $t=0$ using the definition in \eqref{lt}. We then solved \eqref{Cpded} forward in time. As was mentioned in the previous sections, this equation is nonlinear, because in \eqref{ddg} we computed the map $x(p,t)$ using the relation $x = C^*_1(p,t)$. Therefore, we solved it using an iterative scheme similar to \eqref{SchemeF}. Having that map at every iteration we re-interpolated the values of $C_{11}(x,t)$ found and stored during the backward run to the corresponding grid in $p$.

\begin{figure}[!ht]
\begin{center}
\fbox{\includegraphics[width=5 in]{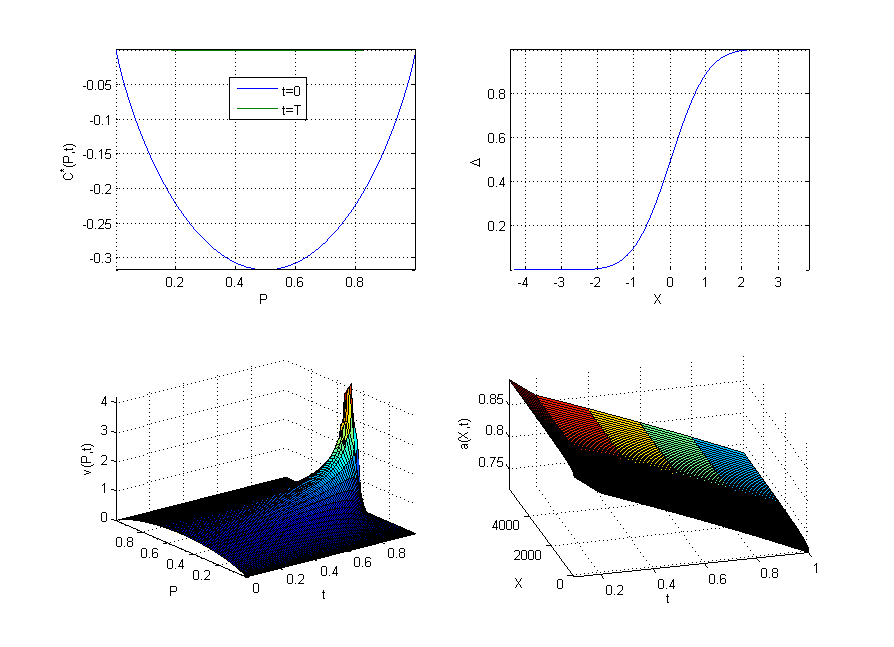}}
\caption{Results of the calculations when the local volatility function $v(p,t)$ is obtained based on the backward recursion.}
\label{Figback}
\end{center}
\end{figure}

In Fig.~\ref {Figback} the computed results are presented in four graphs. The top left graph shows the initial function $c^*(p,0)$ computed using the solution of the backward equation, and the final solution $c^*(p,T)$ computed using \eqref{Cpded}. The top right plot in Fig.~\ref {Figback} demonstrates the map $x \to p(x,t)$ computed using the solution of the backward equation at $t=0$. The bottom right graph represents our local volatility function $a(x,t)$ given in Table.~\eqref {locVol1}. And the bottom left graph shows $v(p,t)$ computed in this experiment.

It is seen that the Legendre transform converges to its terminal value $c^*(p,T) = 0$. Two important points, however, should be taken into account. First, the solution is very sensitive to the accuracy of Gammas. Suppose that the backward equation is solved using a standard, second order accurate in space, scheme. Then the accuracy of Gamma is $O(1)$ which is not sufficient (and could be error-prone) for computing $D(p,t)$. Therefore, in this situation it would be reasonable to use a higher order scheme in $x$ or $p$, for instance a HOC finite difference scheme of the forth order, see, e.g., \cite{chawla2000}.

Second, when $t \to T$ our map $x(p,t)$ degenerates. Indeed, if we fix $x$, at $t=0$ the map 
$x(p,t)$ is regular, i.e. for every $x$ it provides a unique value of $p$. At $t=T$ for the call option for any $x$ the value of $p$ could be either 0 or 1. Therefore, multiple values of $x$ are connected to the same value of $p$. Obviously, the inverse statement is also true. Suppose that at $t=0$ we fix the grid in $p$, and then for every $p$ a unique value of $x(p,0)$ could be found. However, at $t=T$ all these values of $p$ will connect to the value $x=0$, because $x = C^*_1(p,t)$ and $C^*(p,T) = 0$. This degeneracy of the map produces additional problems with computing the accurate values of $v(p,t)$.

Fig.~\ref {map} shows computed plot of $x(p,t)$ obtained in this experiment.
\begin{figure}[!ht]
\begin{center}
\fbox{\includegraphics[width=4 in]{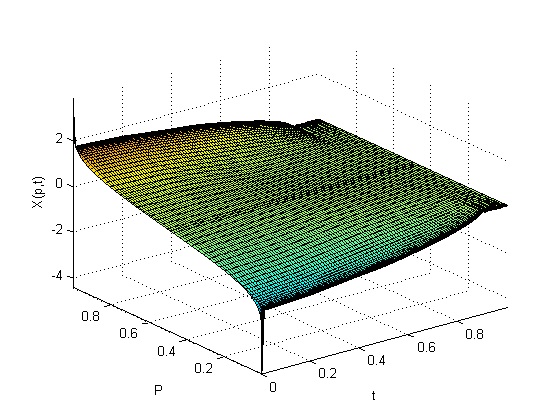}}
\caption{Map $x(p,t)$ computed in our test.}
\label{map}
\end{center}
\end{figure}

\begin{figure}[!ht]
\begin{center}
\fbox{\includegraphics[width=5 in]{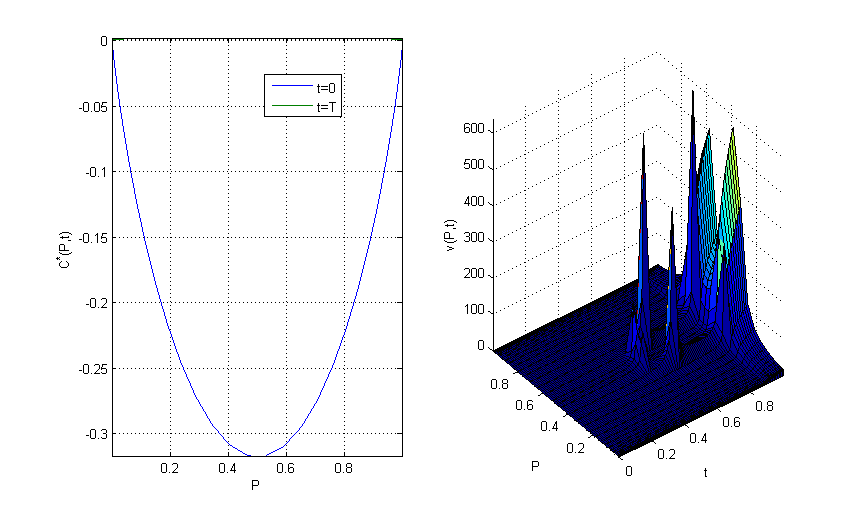}}
\caption{Results of the calculations when the local volatility function $v(p,t)$ is computed at the forward recursion.}
\label{FigForw}
\end{center}
\end{figure}

It is seen that at $t=0$ this plot replicates the top right graph in Fig.~\ref {Figback}, and at $t=T$ it vanishes.

In the second test we solve the problem \eqref{nonl}. Again we first solve the backward problem 
and find the initial values of $c^*(p,0)$ as was described in above. And then using the scheme \eqref{SchemeF} propagate the solution forward in calendar time $t$ up to the maturity $T$.

In Fig.~\ref {FigForw} thus computed results are presented in two graphs similar to the left top and bottom ones in Fig,~\ref {Figback}.

Observe, that the option value at $t=T$ is equal to payoff, hence the option gamma $c_{11}(x,T)$ is the Dirac's delta function. Accordingly, the local delta volatility $v(p,T)$ is also the delta function if the local forward volatility $a(x,t)$ is finite. That brings some technical difficulties to the numerical solution. The computed value of $v(p,t)$ is depicted in the right plot in Fig.~\ref {FigForw}. At every $t$ this function demonstrates the correct behavior by vanishing at $p=0$ and $p=1$, and showing a positive bell-shaped profile at the intermediate values. At $t \to T$ it tends to the delta function. The accuracy of the solution, accordingly, significantly drops down.

Certainly, for our scheme \eqref{SchemeF} the number of iterations to converge depends on the chosen level of tolerance $\varepsilon$. Our experiments show that this number significantly depends on the smoothness of $v(p,t)$ at every time step. As gammas are computed by numerical differentiation, producing a smooth local volatility function similar in the shape to the Dirac's delta function is not simple. As we already mentioned, even if the second order scheme in $p$ is used to find $c^*(p,t)$, gammas are computed with the approximation $O(1)$. Therefore, at every iteration after the solution was computed a polynomial smoothing was applied. In particular, we fitted the solution with the polynomial of the 4th order, and then computed the second derivative of this polynomial. This provides a significant speedup in convergence. 

If gammas were computed by numerical differentiation, at the first time step the number of iterations $k$ was the highest, while when $t$ increases $k$ rapidly decreases. For instance, in the above test at $\varepsilon = 10^{-5}$ the first step in time requires 280 iterations, while starting from the time step 6 it decreases to 5 iterations per one time step. However, still at some values of the model parameters the entire scheme could diverge. Nevertheless, when using the polynomial fit the typical number of iterations was 2-3 at every time step even with $\varepsilon = 10^{-8}$, and the scheme rapidly converged.

We need to underline that despite the dual problem was formulated in \eqref{nonl} subject to the given terminal condition, here we solved this problem assuming that the initial condition at $t=0$ is known from the solution of the backward PDE. We then demonstrated that this solution converges to the correct terminal value of $c^*(p,t)$ within the error of our numerical procedure.  That is definitely because we used the correct initial value. If, however, the initial condition is not known, one can use either the shooting method, or the invariant imbedding method referenced in the above. This, certainly, will significantly increase the computational time.

\section{Conclusions}
This paper deals with a certain type of nonlinear PDE which naturally appears in mathematical finance 
when solving some problems formulated as optimization. As the known examples of these problems we 
considered an optimal consumption problem presented in \cite{merton:71} and then in detail in \cite{Lipton2001} and a dual approach to option pricing invented by \cite{CarrQUSA2014}. We showed that Carr's dual equation can be translated into the nonlinear PDE, and that the nonlinear PDE of \cite{Lipton2001} is a particular case of the former when the volatility function is linear in the underlying variable. Further on, we described some properties of the general nonlinear PDE, provided its analytical solution in the Black-Scholes case, and proposed a new numerical method to solve it which is unconditionally stable and of the second order in both time and space. Stability of the method and approximation are proved by the corresponding Theorem formulated and proved in the paper. We also presented some results of our numerical experiments. All these results are new.

As far as the practical applications of this approach is concerned, as mentioned in \cite{CarrQUSA2014}, for some valuation problems, it may in fact be easier to specify the problem in the dual domain. This could be important, e.g., for doing delta hedge for equity options, or in the FX market where the quotes are expressed in terms of implied volatility as a function of the Black Scholes delta. Also for American options working in the delta space would make it much easier to setup the exercise boundary. Therefore, it does make sense to extend the dual technology, and, perhaps, also including jumps into consideration. From this prospective, these extensions would require new numerical methods for nonlinear PDEs most likely similar to those considered in this paper.

\section*{Acknowledgments}
We are grateful to Peter Carr and Alex Lipton for their comments. They are not responsible for any errors.

\section*{References}

\begin{thebibliography}{12}
\expandafter\ifx\csname natexlab\endcsname\relax\def\natexlab#1{#1}\fi
\expandafter\ifx\csname url\endcsname\relax
  \def\url#1{\texttt{#1}}\fi
\expandafter\ifx\csname urlprefix\endcsname\relax\def\urlprefix{URL }\fi

\bibitem[{Bellman and Wing(1992)}]{BellmanWing1992}
Bellman, R., Wing, G., 1992. An Introduction to Invariant Imbedding. Society
  for Industrial and Applied Mathematics.

\bibitem[{Berman and Plemmons(1994)}]{BermanPlemmons1994}
Berman, A., Plemmons, R., 1994. Nonnegative matrices in mathematical sciences.
  SIAM.

\bibitem[{Carr(2014)}]{CarrQUSA2014}
Carr, P., July 2014. Options as optimizations: A dual approach to derivatives
  pricing. Quant USA. New York.

\bibitem[{Chawla et~al.(2000)Chawla, Al-Zanadi, and Al-Aslab}]{chawla2000}
Chawla, M.~M., Al-Zanadi, M.~A., Al-Aslab, M.~G., 2000. Extended one-step
  time-integration schemes for convection-diffusion equations. Computers and
  Mathematics with Applications 39, 71--84.

\bibitem[{Howison(1995)}]{Howison1995}
Howison, S., 1995. Barrier options. Available at
  \url{https://people.maths.ox.ac.uk/howison/barriers.pdf}.

\bibitem[{{In't Hout} and Foulon(2010)}]{HoutFoulon2010}
{In't Hout}, K.~J., Foulon, S., 2010. {ADI} finite difference schemes for
  option pricing in the {H}eston model with correlation. International journal
  of numerical analysis and modeling 7~(2), 303--320.

\bibitem[{Itkin and Carr(2011)}]{ItkinCarrBarrierR3}
Itkin, A., Carr, P., 2011. Jumps without tears: A new splitting technology for
  barrier options. International Journal of Numerical Analysis and Modeling
  8~(4), 667--704.

\bibitem[{Li et~al.(2010)Li, Lu, Wang, and McCammon}]{LiLuWangMcCammon2010}
Li, B., Lu, B., Wang, Z., McCammon, J., 2010. Solutions to a reduced
  poisson–-nernst–-planck system and determination of reaction rates.
  Physica A 389, 1329--1345.

\bibitem[{Lipton(2001)}]{Lipton2001}
Lipton, A., 2001. Mathematical Methods For Foreign Exchange: A Financial
  Engineer's Approach. World Scientific.

\bibitem[{Merton(1971)}]{merton:71}
Merton, R.~C., December 1971. Optimum consumption and portfolio rules in a
  continuous-time model. Journal of Economic Theory 3~(4), 373--413.

\bibitem[{Roberts and Shipman(1972)}]{sm}
Roberts, S., Shipman, J., 1972. Two-Point Boundary Value Problems: Shooting
  Methods. American Elsevier Pub. Co.

\bibitem[{Starzak(1989)}]{MMCP89}
Starzak, M., 1989. Mathematical methods in chemistry and physcis. Springer, New
  York.

\end{thebibliography}
\newcommand{\noopsort}[1]{} \newcommand{\printfirst}[2]{#1}
  \newcommand{\singleletter}[1]{#1} \newcommand{\switchargs}[2]{#2#1}

\appendix
\section{Proof of Proposition~\protect{\ref {Prop1}}}

Let us remind that the solution of \eqref{nonl} is given by \eqref{Scheme} in the operator form 
\begin{equation} \label{oper}
C^{(k)}(p,t) = e^{\Delta t \mathcal{L}} C^{(k)}(p,t-\Delta t).
\end{equation}
\noindent where
\[ \mathcal{L} = \dfrac{1}{2}[D^{(k-1)}(p,t)+D(p,t-\Delta t)] \partial_{pp}, \qquad
D^{(k)}(p,t) = \dfrac{1}{2}\left[\dfrac{\hat{a}(p,t)}{C^{(k)}_{pp}(p,t)}\right]^2, \]
Consider a discrete analog $A^C_2$ of the operator $\partial_{pp}$ which could be obtained on the grid ${\bf G}(x)$ by using finite-difference approximation \eqref{app2d}. Observe, that $A^C_2$ is the Metzler matrix, see \cite{BermanPlemmons1994}. Indeed, it has all negative elements on the main diagonal, and all nonnegative elements outside of the main diagonal. Also $A_2^C$ is a tridiagonal matrix. 

Now observe that $\Delta t > 0$ and $D^k(p,t) > 0$. Therefore, matrix $M = \dfrac{1}{2}\Delta t  [D^{(k-1)}(p,t)+D(p,t-\Delta t)] A^C_2$ is also the Metzler matrix.

By the properties of the Metzler matrix its exponent is a nonnegative matrix. Therefore, $e^M$ preserves the sign of the vector $C^{(k)}(p,t)$. Also all eigenvalues of the Metzler matrix have a negative real part. Therefore, the spectral norm of the matrix $e^M$ follows
\[ \| e^M \| < 1. \]
Thus, the map $e^M$ is contractual, and hence, \eqref{oper} is unconditional stable.

Now we prove that the matrix $e^M$ is the second order approximation of the operator $e^{\Delta t \mathcal{L}}$. That follows from the fact that matrix $A^C_2$ approximates the operator $\partial_{pp}$ with the second order, i.e. on the non-uniform grid it provides approximation $O(h_i (h_i + h_{i+1})), \ i \in [1,N]$. 

The last point is to prove the convergence of the fixed point Picard iterations. For doing that denote by $\hat{C}(p,t+\Delta t)$ the exact solution of \eqref{nonl1}. According to \eqref{Scheme} it could be represented as 
\begin{equation} \label{exactS}
\hat{C}(p,t+\Delta t) = \mathcal{L} C(p,t), \qquad \mathcal{L} \equiv \exp\left\{\dfrac{1}{2} \Delta t \hat{D}(p,t) \triangle \right\}
\end{equation}
\noindent where 
\[ \hat{D}(p,t) =  \dfrac{1}{2}\left[\dfrac{\hat{a}(p,t)}{\hat{C}_{11}(p,t)}\right]^2, \]
Observe, that at every $k$ $\mathcal{L}^{(k)}$ is a linear bounded operator because it is Lipschitz
\[ \Big\| \mathcal{L}^{(k)} C_1(p,t) - \mathcal{L}^{(k)} C_2(p,t)\Big\| \le \Big\| \mathcal{L}^{(k)} \Big\| \|C_1(p,t) - C_2(p,t)\| \le \|C_1(p,t) - C_2(p,t)\| \]
\noindent since as it was shown earlier in the the spectral norm $\|\mathcal{L}^{(k)}\| \le 1$.

Then  by the mean-value theorem for operators we have
\begin{align*}
\Big\| &C^{(k)}(p,t+\Delta t) - \hat{C}(p,t+\Delta t) \Big\| = \Big\| \mathcal{L}^{(k-1)} C(p,t)  - \hat{\mathcal{L}} C(p,t) \Big\| \\
&\le \| \mathbb{D}(\mathcal{L})(\xi^{(k}))\| \Big\| C^{(k-1)}(p,t+\Delta t) - \hat{C}(p,t) \Big\|
\end{align*}
\noindent where $\xi^{(k)}$ is a convex combination of $C^{(k-1)}(p,t+\Delta t)$ and $C^{(k)}(p,t+\Delta t)$, and $\mathbb{D}$ denotes the Fr\'echet derivative of the operator 
$\mathcal{L}$ at the space of all bounded linear operators, see, e.g., \cite{LiLuWangMcCammon2010}.

To compute the norm of the Fr\'echet derivative recall that by definition 
\[ \| \mathbb{D}(\mathcal{L}) \| = \sup_{u \ne 0} \dfrac{\| \mathbb{D}(\mathcal{L})(u) \|}
{\| u \|}\]
If $u = (u_1,...,u_m) \in [L^\infty(-\infty,0)]^m$, then 
\[ 
\mathbb{D}(\mathcal{L})(u) = \fp{\mathcal{L}(C + t u)}{t}\Big|_{t=0}, \]
\noindent and, therefore
\begin{equation} \label{funcDer}
\| \mathbb{D}(\mathcal{L})(u) \|  = \hat{a}^2 \Delta t  \Big\| \dfrac{\mathcal{L}(C)}{C^3_{pp}} u_{pp} \Delta \Big\| \| u\|^{-1} \le \hat{a}^2 \Delta t \dfrac{1}{\| \Delta\|}.
\end{equation}
It is well-known that, \cite{MMCP89} 
\[ \| \Delta \| = \dfrac{4}{h^2}\sin^2 \dfrac{i \pi}{2(N+1)}, \quad i \in [1,...,N] \]
\noindent where $N$ is the size of the matrix $\Delta$. Therefore, $\Delta t$ and $h$ could be easily chosen such that $\| \mathbb{D}(\mathcal{L})(u) \| < 1$. It is important, that thus chosen values are independent, and, therefore, the scheme remains to be almost" unconditionally stable. 

To explain what almost" means let us compare our condition with the familiar stability condition for the Euler explicit finite-difference scheme, which reads
\begin{equation} \label{stabCond}
2 \nu \Delta t \le h^2 
\end{equation}
\noindent with $\nu$ being some diffusion coefficient. As it can be seen, the latter condition is restrictive, because it sets the upper limit on the time step given the space step $h$. In contrast, our condition reads
\[  \hat{a}^2 \Delta t h^2 \left[4\sin^2 \dfrac{i \pi}{2(N+1)}\right]^{-1} < 1 \]
In the worst case scenario this could be re-written as 
\[ \hat{a}^2 \Delta t h^4 < h^2   \]
So comparing this with \eqref{stabCond} one can see, that the stability is conditional. However it is considerably relaxed as compared with that in \eqref{stabCond} because of the presence of the multiplier $h^4$ in the left hand side. Indeed, we need $\Delta t < 1/(h \hat{a})^2$ where the right hands part is usually a huge value unless the local volatility $\hat{a}(p,t)$ is very very high which is impractical. Therefore, for any practical application theoretical restrictions on $\Delta t$ are not important.
$\blacksquare$

\end{document}